\documentclass[11pt]{article}
\usepackage{amsmath}
\usepackage{amsthm}
\usepackage{amssymb}
\usepackage{latexsym}
\numberwithin{equation}{section}
\newtheorem{theo}{Theorem}[section]
\newtheorem{cor}[theo]{Corollary}
\newtheorem{lemma}[theo]{Lemma}
\newtheorem{prop}[theo]{Proposition}
\theoremstyle{definition}
\newtheorem{defi}[theo]{Definition}

\newtheorem{rem}[theo]{Remark}
\newtheorem{nota}[theo]{Notation}

\newcommand{\F}{{\mathbb F}}
\renewcommand{\P}{{\mathbb P}}
\newcommand{\FFF}{\mathbb{F}_q^{3\times 3}}

\newcommand{\cC}{{\mathcal C}}
\newcommand{\cD}{{\mathcal D}}
\newcommand{\cF}{{\mathcal F}}

\newcommand{\cI}{{\mathcal I}}

\newcommand{\cS}{{\mathcal S}}

\newcommand{\rk}{\mbox{${\rm rk}$\,}}
\newcommand{\Prob}{\mbox{\rm Prob}\,}
\renewcommand{\mod}{\mbox{\rm mod}\,}
\newcommand{\Orbt}{\mbox{${\rm Orb}_\tau$}}

\newcommand{\T}{\mbox{$\!^{\sf T}$}}
\newcommand{\subspace}[1]{\mbox{$\langle{#1}\rangle$}}
\newcommand{\GL}{\mathrm{GL}}
\newcommand{\Stab}{\mathrm{Stab}}
\newcommand{\bra}[1]{\mbox{$^{[#1]}$}}
\newcommand{\Gaussian}[2]{\genfrac{[}{]}{0pt}{1}{#1}{#2}}
\newcounter{alp}
\newcounter{ara}
\newcounter{rom}
\newenvironment{romanlist}{\begin{list}{(\roman{rom})\hfill}{\usecounter{rom}
     \topsep0ex \labelwidth1.7em \leftmargin1.7em \labelsep0cm
     \rightmargin0cm \parsep0ex \itemsep.4ex
     \partopsep1ex}}{\end{list}}
\newenvironment{alphalist}{\begin{list}{(\alph{alp})\hfill}{\usecounter{alp}
     \topsep0ex \labelwidth.6cm \leftmargin.6cm \labelsep0cm
     \rightmargin0cm \parsep0ex \itemsep0ex
     \partopsep0ex}}{\end{list}}
\newenvironment{arabiclist}{\begin{list}{(\arabic{ara})\hfill}{\usecounter{ara}
     \topsep0ex \labelwidth1.4em \leftmargin1.4em \labelsep0cm
     \rightmargin0cm \parsep0ex \itemsep0ex
     \partopsep1.6ex}}{\end{list}}

\begin{document}
\title{On the Sparseness of Certain MRD Codes}
\date{\today}
\author{Heide Gluesing-Luerssen\footnote{HGL was partially supported by the grant \#422479 from the Simons Foundation.
Department of Mathematics, University of Kentucky, Lexington KY 40506-0027, USA;
heide.gl@uky.edu.}}

\maketitle

{\bf Abstract:}
We determine the proportion of $[3\times 3;3]$-MRD codes over~$\F_q$ within the space of all $3$-dimensional $3\times3$-rank-metric
codes over the same field.
This shows that for these parameters MRD codes are sparse in the sense that the proportion tends to~$0$ as~$q\rightarrow\infty$.
This is so far the only parameter case for which MRD codes are known to be sparse.
The computation is accomplished by reducing the space of all such rank-metric codes to a space of specific bases and subsequently
making use of a result by Menichetti (1973) on 3-dimensional semifields.

\smallskip


\smallskip


\section{Introduction}

Rank-metric codes play a crucial role for error correction in single-source networks with adversarial noise; see \cite{SKK08,SiKsch09}.
For this reason they have been studied in great detail during the last decade (and even longer) in the context of coding theory.
The main focus has been on MRD codes, that is, rank-metric codes with the largest possible dimension for the given rank distance.
Such codes do exist for all possible parameters,  and constructions have been provided independently by Delsarte~\cite{Del78} and later Gabidulin~\cite{Gab85}.
The latter codes are even $\F_{q^m}$-linear when identifying the ambient matrix space $\F_q^{m\times n}$ with $\F_{q^m}^n$.
Further constructions of MRD codes (not $\F_{q^m}$-linear in general) have been found and studied by Sheekey~\cite{Sh16}, Cossidente et al.~\cite{CMP16}, Csajb{\'o}k et al.~\cite{CMPZ17},
Lunardon et al.~\cite{LTZ18}, Trombetti/Zhou~\cite{TrZh19} and others.

In this paper we focus on the proportion of $[m\times n;\delta]_q$-MRD codes, that is, the number of such MRD codes compared to the
number of all $m\times n$-rank-metric codes over~$\F_q$ of the same dimension.
Bounds on this proportion have been given already  by Byrne/Ravagnani~\cite{ByRa18} and Antrobus/Gluesing-Luerssen~\cite{AGL19}.
More precisely, in \cite[Cor.~6.2]{ByRa18} it is shown that, for all parameter sets and all field sizes, this proportion is always at most 1/2, while
in~\cite[Thm.~VII.6]{AGL19} a much smaller upper bound is derived via entirely different methods.
Both results leave the question open whether for some parameters $(m,n,\delta)$ the asymptotic proportion is zero as $q\rightarrow\infty$, in which case the family of $[m\times n;\delta]$-MRD
codes is called sparse.
Besides the trivial case $\delta=1$, the only known case for the exact asymptotic proportion is $(m,n,\delta)=(m,2,2)$.
In that case the limit is nonzero and has been determined in \cite[Cor.~VII.5]{AGL19}.
Further details are provided in the next section after introducing the necessary notation and terminology.

We will explicitly compute the proportion of $[3\times3;3]_q$-MRD codes.
We will see that this proportion approaches~$0$ as $q\rightarrow\infty$, and thus $[3\times3;3]$-MRD codes are sparse.
This is the only parameter case known so far for which MRD codes are sparse.
The main part of our proof consists of using a characterization of $3$-dimensional semifields over finite fields by Menichetti~\cite{Men73}.
This is also the reason why our results only apply to the parameter set $(m,n,\delta)=(3,3,3)$.

The proportion will be obtained in two steps, formulated in Theorem~\ref{T-Reduction} and Theorem~\ref{T-FinalCount} in the next section.
In Theorem~\ref{T-Reduction} (proven in Section~\ref{S-Thm1}) we reduce the given problem to the problem of counting a certain set of matrix triples.
This set consists of bases of MRD codes, and in fact all MRD codes up to isometry are covered by this set.
However, the set may also contain isometric codes, but since we aim at counting \emph{all} codes there is no need or use in studying isometry
(and in fact, this has been done by Menichetti in~\cite{Men77}).
In Theorem~\ref{T-FinalCount} (established in Section~\ref{S-Thm2}) we observe that the set of matrix triples has been parametrized by Menichetti in~\cite{Men73}.
A closer look at those results will provide us with the exact count.
While we formulate the results and proofs in Sections~\ref{S-Prelim} --~\ref{S-Thm2} in a  purely matrix-theoretical language, the main result in Section~\ref{S-Thm2}, Theorem~\ref{T-MatSigma},  by Menichetti
is best understood in the language of semifields.
For this reason, we will give in Section~\ref{S-Semifields} an account of the main ideas in~\cite{Men73} leading to Theorem~\ref{T-MatSigma}.
We believe it is worth to present that beautiful material in a slightly reorganized and more extended form that suits the current interest in MRD codes.
Among other things, this will allow us to detect the commutative non-associative semifields within our set of matrix triples.

\section{Preliminaries and Statement of Main Results}\label{S-Prelim}

Throughout, $\F_q$ denotes a finite field of order~$q$.
We endow the  matrix space $\F_q^{m\times n}$ with the \emph{rank metric}, defined as $d(A,B)=\rk(A-B)$.
A \emph{rank-metric code} is an $\F_q$-linear subspace of the metric space $(\F_q^{m\times n},d)$.
The \emph{rank distance} of a rank-metric code~$C\subseteq\F_q^{m\times n}$ is defined as $d_R(C):=\min\{\rk(M)\mid M\in C\backslash\{0\}\}$.
An $[m\times n, k; \delta]_q$-code is a rank-metric code in $\F_q^{m\times n}$ of dimension~$k$ and rank distance~$\delta$.
Assuming (without loss of generality) that $n\leq m$, the \emph{Singleton bound} tells us that the dimension~$k$ of an $[m\times n,k;\delta]_q$-code is at most $m(n-\delta+1)$, and
codes attaining this bound are called \emph{MRD codes} (maximum rank-distance codes), denoted as $[m\times n;\delta]$-MRD codes.
It is well known that MRD codes exist for all parameters $(m,n,\delta)$ and all fields~$\F_q$.
For further details on rank-metric codes we refer to the vast literature.

In this paper we focus on the proportion of MRD codes in the following sense.

\begin{defi}\label{D-Spaces}
Let $1\leq\delta\leq n\leq m$. Consider the spaces
\[
       T_q^{m,n,\delta}=\{\cC\subseteq\F_q^{m\times n}\mid \dim(\cC)=m(n-\delta+1)\}\ \text{ and }\
       \hat{T}_q^{m,n,\delta}=\{\cC\in T_q^{m,n,\delta}\mid d_R(\cC)=\delta\}.
\]
Thus $\hat{T}_q^{m,n,\delta}$ is the set of $[m\times n;\delta]$-MRD codes.
The fraction $|\hat{T}_q^{m,n,\delta}|/|T_q^{m,n,\delta}|$ is called the \emph{proportion of MRD-codes} (within the space of all $m(n-\delta+1)$-dimensional subspaces of $\F_q^{m\times n}$).
\begin{alphalist}
\item The family of $[m\times n;\delta]$-MRD codes is called \emph{sparse} if $\lim_{q\rightarrow\infty}|\hat{T}_q^{m,n,\delta}|/|T_q^{m,n,\delta}|=0$.
\item The family of $[m\times n;\delta]$-MRD codes is called \emph{generic} if $\lim_{q\rightarrow\infty}|\hat{T}_q^{m,n,\delta}|/|T_q^{m,n,\delta}|=1$.
\end{alphalist}
\end{defi}

The  above defined proportion is related to the probability that~$m(n-\delta+1)$ randomly and independently chosen matrices in~$\F_q^{m\times n}$
(endowed with the uniform distribution) generate an MRD code.
Indeed, in \cite[Prop.~VI.2]{AGL19} it is shown that the two quantities differ by a factor which tends to~$1$ as $q\rightarrow\infty$.
See also Remark~\ref{R-Proba} in the next section.

We summarize the results known so far about the proportion of MRD codes.

\begin{rem}
\begin{alphalist}
\item It follows from \cite[Thm.~VII.6]{AGL19} that the asymptotic proportion satisfies
        \[
           \lim_{q\rightarrow\infty}\frac{|\hat{T}_q^{m,n,\delta}|}{|T_q^{m,n,\delta}|}\leq\bigg(\sum_{j=0}^m(-1)^j/j!\bigg)^{(\delta-1)(n-\delta+1)}
        \]
        (if the limit exists).
        Hence $[m\times n;\delta]$-MRD codes are not generic if $\delta>1$.
        For $\delta=1$, the full matrix space is the only $[m\times n;1]$-MRD code and thus trivially generic.
\item In \cite[Cor.~VII.5]{AGL19} it has been shown that for $(m,n,\delta)=(m,2,2)$ the proportion attains the upper bound given in~(a).
        Hence, in this case MRD codes are neither sparse nor generic.
\item Certain maximal Ferrers diagram codes are generic; see \cite[Thm.~VI.8]{AGL19}.
        These are rank-metric codes where all matrices are zero in prespecified positions and satisfy a generalized Singleton bound.
        For instance, genericity holds for upper triangular matrices.
        We refer to~\cite{AGL19} for further details.
\item Identifying $\F_q^m$ with the field extension $\F_{q^m}$ we obtain an $\F_q$-vector space isomorphism between the matrix space $\F_q^{m\times n}$ and~$\F_{q^m}^n$.
        This allows us to consider rank-metric codes in the latter space, in which case one may also require the codes to be $\F_{q^m}$-linear.
        In~\cite{NHTRR18} it is shown that $\F_{q^m}$-linear MRD codes are generic within the class of all $\F_{q^m}$-linear rank-metric codes of the same dimension.
\end{alphalist}
\end{rem}

In this paper we add another case to the above list.
It is in stark contrast to those earlier results.
Indeed, we will prove that  $[3\times3;3]$-MRD codes are sparse in the sense of Definition~\ref{D-Spaces}.

From now on we restrict ourselves to $n=m=\delta=3$ and set
\begin{equation}\label{e-TThat}
   T_q=\{\cC\subseteq\FFF\mid \dim(\cC)=3\}\ \text{ and }\
   \hat{T}_q=\{\cC\in T_q\mid d_R(\cC)=3\}.
\end{equation}
We will prove the two theorems below.
The first result reduces the count of $|\hat{T}_q|/|T_q|$ to determining the cardinality of the set~$\cS$ defined below, which consists of very specific matrix triples generating MRD codes.
The second result provides the actual cardinality of~$\cS$.

The following notation is needed.
Set
\begin{equation}\label{e-Irr}
   \cI:=\{f\in\F_q[x]\mid f\text{ monic of degree~$3$ and irreducible}\},
\end{equation}
and for $f=x^3-cx^2-bx-a\in\F_q[x]$ define its companion matrix as
\begin{equation}\label{e-Cf}
   C_f:=\begin{pmatrix}0&0&a\\1&0&b\\0&1&c\end{pmatrix}.
\end{equation}
For any matrix $Z\in\FFF$ we let $z_{ij}$ be its entry at position~$(i,j)$.
Finally, $\subspace{\,\cdot\,}$ denotes the subspace of $\FFF$ generated by the listed matrices.

\begin{theo}\label{T-Reduction}
Define the set
\[
  \cS=\{(I,C_f,Z)\in(\FFF)^3\mid f\in\cI,\, z_{11}=z_{21}=0,\,z_{31}=1,\,\subspace{I,C_f,Z}\text{ is MRD}\}.
\]
Then the proportion of all $[3\times3;3]$-MRD codes is given by
\[
  \frac{|\hat{T}_q|}{|T_q|}=|\cS|\frac{(q-1)(q^3-1)(q^3-q)^2(q^3-q^2)^2}{(q^7-1)(q^9-1)(q^9-q)}.
\]
\end{theo}

From the fraction it is clear that $\lim_{q\rightarrow\infty}|\hat{T}_q|/|T_q|=0$ if $|\cS|$ is at most of order~$q^8$.
Our second result shows that~$|\cS|$ is actually of order~$q^6$.

\begin{theo}\label{T-FinalCount}
The set~$\cS$ has cardinality $\frac{q^3-q}{3}(q^3-q^2-q-1)$.
As a consequence, the proportion of $[3\times3;3]$-MRD codes is
\[
   \frac{|\hat{T}_q|}{|T_q|}= \frac{(q-1)(q^3-1)(q^3-q)^3(q^3-q^2)^2(q^3-q^2-q-1)}{3(q^7-1)(q^9-1)(q^9-q)}
\]
and $\lim_{q\rightarrow\infty}|\hat{T}_q|/|T_q|=0$.
That is, $[3\times3;3]$-MRD codes are sparse.
\end{theo}

Note that~$|T_q|$ is simply the number of $3$-dimensional subspaces in~$\F_q^9$, thus $|T_q|=\Gaussian{9}{3}_q$ (the Gaussian coefficient).
As a consequence, we can determine the absolute number of $[3\times3;3]$-MRD codes.
One obtains for instance $|\hat{T}_2|=192,\,|\hat{T}_3|=870,\!912$, and $|\hat{T}_5|=4,\!512,\!000,\!000$, which coincides with the numerical results by Sheekey~\cite[Sec.~5]{Sh19}.

Theorem~\ref{T-Reduction} will be proven in the next section, while Theorem~\ref{T-FinalCount} will be established in Section~\ref{S-Thm2}.

\section{Reduction Step: Proof of Theorem~\ref{T-Reduction}}\label{S-Thm1}
In order to prove Theorem~\ref{T-Reduction} we replace the subspaces in~$\hat{T}_q$ and $T_q$, see~\eqref{e-TThat}, by ordered bases.
This will be done in several steps.

\begin{prop}\label{P-BasisMatrix}
Define the spaces
\[
  V_q=\{(I,A_2,A_3)\in(\FFF)^3\mid \dim\subspace{I,A_2,A_3}=3\},\quad
   \hat{V}_q=\{(I,A_2,A_3)\in V_q\mid d_R\subspace{I,A_2,A_3}=3\}.
\]
Then
\[
  \frac{|\hat{T}_q|}{|T_q|}=\frac{|\hat{V}_q|}{|V_q|}\cdot \frac{\prod_{i=0}^2(q^3-q^i)}{q^9-1}.
\]
As a consequence, $\lim_{q\rightarrow\infty}|\hat{T}_q|/|T_q|=\lim_{q\rightarrow\infty}|\hat{V}_q|/|V_q|$ if the limit exists.
\end{prop}

\begin{proof}
Consider the matrix spaces
\begin{align*}
     W_q&=\{(A_1,A_2,A_3)\in(\FFF)^3\mid \dim\subspace{A_1,A_2,A_3}=3\},\\
   \hat{W}_q&=\{(A_1,A_2,A_3)\in W_q\mid d_R\subspace{A_1,A_2,A_3}=3\}.
\end{align*}
Every code in $T_q$ (and thus in $\hat{T}_q$) has $\prod_{j=0}^2(q^3-q^j)$ ordered bases.
Since the matrix triples in~$W_q$ and~$\hat{W}_q$ form all ordered bases of the codes in~$T_q$ and~$\hat{T}_q$, respectively,
we obtain $\frac{|\hat{T}_q|}{|T_q|}=\frac{|\hat{W}_q|}{|W_q|}$.
Next, we clearly have
\begin{equation}\label{e-WV}
   |W_q|=\prod_{i=0}^2(q^9-q^i)\ \text{ and }\ |V_q|=\prod_{i=1}^2(q^9-q^i).
\end{equation}
In order to relate the cardinalities of $\hat{W}_q$ and $\hat{V}_q$, consider the map
\[
  \varphi:\hat{W}_q\longrightarrow \hat{V}_q,\quad (A_1,A_2,A_3)\longmapsto (I,A_1^{-1}A_2,A_1^{-1}A_3).
\]
This map is well-defined because every matrix of a triple $(A_1,A_2,A_3)\in\hat{W}_q$ is invertible thanks to $d_R\subspace{A_1,A_2,A_3}=3$,
and clearly $d_R\subspace{A_1,A_2,A_3}=d_R\subspace{I,A_1^{-1}A_2,A_1^{-1}A_3}$.
Furthermore,~$\varphi$ is surjective and each fiber is of the form
$\varphi^{-1}(I,B_2,B_3)=\{(A_1,A_1B_2,A_1B_3)\mid A_1\in\GL_3(q)\}$.
All of this shows that $|\hat{W}_q|=|\hat{V}_q|\cdot|\GL_3(q)|=|\hat{V}_q|\prod_{i=0}^2(q^3-q^i)$.
Using~\eqref{e-WV} we obtain
\[
  \frac{|\hat{T}_q|}{|T_q|}=\frac{|\hat{W}_q|}{|W_q|}=\frac{|\hat{V}_q|\prod_{i=0}^2(q^3-q^i)}{\prod_{i=0}^2(q^9-q^i)}\cdot\frac{\prod_{i=1}^2(q^9-q^i)}{|V_q|}
  =\frac{|\hat{V}_q|}{|V_q|}\cdot \frac{\prod_{i=0}^2(q^3-q^i)}{q^9-1}.
  \qedhere
\]
\end{proof}

\begin{rem}\label{R-Proba}
The last result can be interpreted in terms of the probability that we obtain an MRD code when drawing random matrices.
Indeed, fix the probability distribution on $\FFF$ where
all entries of a matrix $A=(a_{ij})\in\FFF$  are independent and uniformly distributed.
That is, $\Prob(a_{ij}=\alpha)=q^{-1}$ for all $(i,j)$ and all $\alpha\in\F_q$.
Let
\[
    P_q=\Prob\big(\subspace{A_1,A_2,A_3} \text{ is a $[3\times3;3]$-MRD code}\big),
\]
where $A_1,A_2,A_3\in\FFF$ are \emph{randomly chosen matrices}, i.e., chosen independently and randomly
according to the above distribution.
Then the identity $|\hat{T}_q|/|T_q|=|\hat{W}_q|/|W_q|$ of the previous proof tells us that $P_q$ is exactly the proportion of MRD codes.
Furthermore, Proposition~\ref{P-BasisMatrix} states that up to the factor $\prod_{i=0}^2(q^3-q^i)/(q^9-1)$, this proportion equals the probability
$P_q':=\Prob\big(\subspace{I,A_2,A_3} \text{ is a $[3\times3;3]$-MRD code}\big)$ when choosing $A_2,A_3$ at random.
The factor is less than~$1$, i.e.,~$P_q<P_q'$, which simply reflects the fact that in the latter drawing one matrix is already invertible.
This probabilistic interpretation can be applied accordingly to all further results of this section.
\end{rem}

So far we have reduced the problem of determining the proportion $|\hat{T}_q|/|T_q|$ to determining the cardinality $|\hat{V}_q|$.
In the next step we will make use of a similarity action in order to bring the matrix~$A_2$ of the triples $(I,A_2,A_3)\in\hat{V}_q$ into companion form.
Using the description
\begin{equation}\label{e-hatVq}
  \hat{V}_q=\{(I,A_2,A_3)\in(\FFF)^3\mid \det(x_1I+x_2A_2+x_3A_3)\neq0\text{ for all }(x_1,x_2,x_3)\in\F_q^3\setminus0\}
\end{equation}
and the rational canonical form for matrices in $\FFF$ we observe the following fact.

\begin{rem}\label{R-CharPoly}
Denote by $\chi_A\in\F_q[x]$ the characteristic polynomial of the matrix~$A\in\FFF$.
Then for every triple $(I,A_2,A_3)\in\hat{V}_q$ the characteristic polynomial $f:=\chi_{A_2}\in\F_q[x]$ is irreducible and thus
$A_2$ is similar to the companion matrix~$C_f$; see~\eqref{e-Cf}.
\end{rem}

We need the following result about the stabilizer under the conjugation action.
The first part about $\Stab(C_f)$ can be found in \cite{GoCa93}, which in turn goes back to \cite[\S~217]{Dick58}, while the
second part is an obvious consequence.
One should also note that for every nonzero $s\in\F_q^3$ the matrix $(s\,|\,C_fs\,|\,C_f^2s)$ appearing below is indeed invertible.
This follows from the fact that $\subspace{I,C_f,C_f^2}$ is a field whenever~$f\in\F_q[x]$ is irreducible.
In other words, it is a $[3\times3;3]$-MRD code.

\begin{prop}[\mbox{\cite[Cor.~2 and Cor.~3]{GoCa93}}]\label{P-StabCompMat}
Let~$\cI$ be as in~\eqref{e-Irr} and set $\cD=\{A\in\FFF\mid \chi_A\in\cI\}$.
Denote by $\Stab(A)$ the stabilizer of~$A\in\cD$ w.r.t.\  the conjugation action
\[
   \GL_3(q)\times\cD\longrightarrow\cD,\ (S,A)\longmapsto S^{-1}AS.
\]
Then for any companion matrix $C_f\in\cD$ we have
\[
   \Stab(C_f)=\big\{(s\,|\,C_fs\,|\,C_f^2s)\,\big|\, s\in\F_q^3\setminus0\big\}.
\]
As a consequence, $|\Stab(A)|=q^3-1$ for all $A\in\cD$.
\end{prop}

In our next step toward the proof of Theorem~\ref{T-Reduction} we will transform the triples in~$V_q$ into the form $(I,C_f,Z)$ via an according group action.
For a precise count of the resulting orbits the following lemma is needed.

\begin{lemma}\label{L-StabCfZ}
Let $f\in\cI$ and $X\in\FFF\setminus\subspace{I,C_f,C_f^2}$. Furthermore, let $S\in\Stab(C_f)$. Then
\[
  SX=XS\Longrightarrow S=\alpha I\text{ for some }\alpha\in\F_q^*.
\]
\end{lemma}

\begin{proof}
Write $X=(x_{ij})_{i,j=1,2,3}$ and $f=x^3-cx^2-bx-a$. Then the matrices $C_f$ and~$C_f^2$ are given by
\[
   C_f=\begin{pmatrix}0&0&a\\1&0&b\\0&1&c\end{pmatrix},\
   C_f^2=\begin{pmatrix}0&a&ac\\0&b&a+bc\\1&c&b+c^2\end{pmatrix},
\]
and thus $Z:=X-x_{11}I-x_{21}C_f-x_{31}C_f^2$ is of the form
\[
  Z=\begin{pmatrix}0&z_{12}&z_{13}\\0&z_{22}&z_{23}\\0&z_{32}&z_{33}\end{pmatrix}.
\]
The assumptions imply $Z\neq0$ and $SZ=ZS$.

Let now $S=(s\,|\,C_fs\,|\,C_f^2s)$ where $s=(s_1,s_2,s_3)\T\in\F_q^3\setminus0$; see Proposition~\ref{P-StabCompMat}.
Hence
\[
   S=\begin{pmatrix}s_1&as_3&acs_3+as_2\\s_2&bs_3+s_1&bs_2+(bc+a)s_3\\s_3&cs_3+s_2&cs_2+(c^2+b)s_3+s_1\end{pmatrix}.
\]
We have to show that~$S$ is a multiple of the identity matrix, which means $s_2=s_3=0$.
In other words we have to prove that the solution space of the linear system $SZ-ZS=0$ is given by $\{(s_1,0,0)\mid s_1\in\F_q\}$.
This is now a matter of basic linear algebra.
Indeed, the identity $SZ-ZS=0$ yields a system of~$9$ equations taking the form
\begin{equation}\label{e-sM}
   M\begin{pmatrix}s_1\\s_2\\s_3\end{pmatrix}=0,
\end{equation}
where
\[
 M=\begin{pmatrix}
0 & z_{12} & z_{13} \\
0 & z_{22} & z_{23} \\
0 & z_{32} & z_{33} \\
0 & -a z_{32} + z_{13} & -a c z_{32} + b z_{12} + c z_{13} - a z_{22} \\
0 & -b z_{32} - z_{12} + z_{23} & c z_{23} - {\left(b c + a\right)} z_{32} \\
0 & -c z_{32} - z_{22} + z_{33} & -c^{2} z_{32} - c z_{22} + c z_{33} - z_{12} \\
0 & b z_{12} + c z_{13} - a z_{33} & -a c z_{33} + {\left(b c + a\right)} z_{12} + {\left(c^{2} + b\right)} z_{13} - a z_{23} \\
0 & b z_{22} + c z_{23} - b z_{33} - z_{13} & c^{2} z_{23} + {\left(b c + a\right)} z_{22} - {\left(b c + a\right)} z_{33} \\
0 & b z_{32} - z_{23} & -c z_{23} + {\left(b c + a\right)} z_{32} - z_{13}
\end{pmatrix}.
\]
Thus we have to show that the last two columns of~$M$ are linearly independent.
This is clearly the case if either $\rk(Z)=2$ or $(z_{12},z_{22},z_{32})=(0,0,0)\neq(z_{13},z_{23},z_{33})$.
Hence it remains to consider the case
\[
    Z=\begin{pmatrix}0&z_{12}&tz_{12}\\0&z_{22}&tz_{22}\\0&z_{32}&tz_{32}\end{pmatrix}\text{ for some $t\in\F_q$ and where }
    (z_{12},z_{22},z_{32})\neq0.
\]
Now the last two columns of~$M$ take the form

\[
\tilde{M}(t):=\begin{pmatrix}
 z_{12} & t z_{12} \\
 z_{22} & t z_{22} \\
 z_{32} & t z_{32} \\
 t z_{12} - a z_{32} & -a c z_{32} + {\left(c t + b\right)} z_{12} - a z_{22} \\
 t z_{22} - b z_{32} - z_{12} & c t z_{22} - {\left(b c + a\right)} z_{32} \\
 -{\left(c - t\right)} z_{32} - z_{22} & -c z_{22} - {\left(c^{2} - c t\right)} z_{32} - z_{12} \\
 -a t z_{32} + {\left(c t + b\right)} z_{12} & -a c t z_{32} - a t z_{22} + {\left(b c + {\left(c^{2} + b\right)} t + a\right)} z_{12} \\
 -b t z_{32} - t z_{12} + {\left(c t + b\right)} z_{22} & -{\left(b c + a\right)} t z_{32} + {\left(c^{2} t + b c + a\right)} z_{22} \\
 -t z_{22} + b z_{32} & -c t z_{22} - t z_{12} + {\left(b c + a\right)} z_{32}
\end{pmatrix}.
\]
One computes
\[
   \tilde{M}(t)\begin{pmatrix}1&-t\\0&1\end{pmatrix}=\begin{pmatrix}z_{12}&0\\z_{22}&0\\z_{32}&0\\ \ast&\hat{M}(t){\scriptstyle\begin{pmatrix}z_{12}\\z_{22}\\z_{32}\end{pmatrix}}\end{pmatrix},
\]
where
\[
  \hat{M}(t)=
\begin{pmatrix}
c t - t^{2} \!+\! b & -a & -a c \!+\! a t \\
t & c t - t^{2} & -b c \!+\! b t - a \\
-1 & -c \!+\! t & -c^{2} \!+\! 2 \, c t - t^{2} \\
c^{2} t - c t^{2} \!+\! b c \!+\! a & -a t & -a c t \!+\! a t^{2} \\
t^{2} & -c t^{2} \!+\! b c \!+\! {\left(c^{2} - b\right)} t \!+\! a & b t^{2} - {\left(b c \!+\! a\right)} t \\
-t & -c t \!+\! t^{2} & b c - b t \!+\! a
\end{pmatrix}.
\]
Furthermore,
\[
  \begin{pmatrix}
1 & 0 & 0 & 0 & 0 & 0 \\
0 & 1 & t & 0 & 0 & 0 \\
0 & 0 & 1 & 0 & 0 & 0 \\
-t & 0 & 0 & 1 & 0 & 0 \\
0 & 0 & 0 & 0 & 1 & t \\
0 & 0 & -t & 0 & 0 & 1
\end{pmatrix}\hat{M}(t)=\begin{pmatrix}
c t - t^{2} + b & -a & -a c + a t \\
0 & 0 & -g(t) \\
-1 & -c + t & -c^{2} + 2 \, c t - t^{2} \\
g(t)& 0 & 0 \\
0 & g(t) & 0 \\
0 & 0 &g(t)
\end{pmatrix},
\]
where $g(t)=t^3-2ct^2+(c^2-b)t+bc+a$.
Since $g(t)=-f(c-t)$ and~$f$ is irreducible, we conclude $g(t)\neq0$ and thus $\rk(\hat{M}(t))=3$ for all $t\in\F_q$.
Using that $(z_{12},z_{22},z_{32})\neq0$, we arrive at
$\rk(\tilde{M}(t))=2$ for all $t\in\F_q$.
All of this shows that~\eqref{e-sM} has solution space $\{(s_1,0,0)\mid s_1\in\F_q\}$, and this concludes the proof.
\end{proof}

Now we are ready to prove the following result, which lets us reduce the count of MRD codes to those where one basis matrix is in companion form.
Recall the sets~$V_q$ and~$\hat{V}_q$ from Proposition~\ref{P-BasisMatrix}.

\begin{theo}\label{T-MatComp}
Consider the matrix sets
\begin{align*}
  X_q&=\{(I,C_f,Z)\in(\FFF)^3\mid f\in\cI,\,\dim\subspace{I,C_f,Z}=3\},\\[.5ex]
  \hat{X}_q&=\{(I,C_f,Z)\in X_q\mid d_R\subspace{I,C_f,Z}=3\}.
\end{align*}
Then
\[
  \frac{|\hat{V}_q|}{|V_q|}=\frac{|\hat{X}_q|}{|X_q|}\cdot\frac{1}{3}\frac{(q^3-q)^2(q^3-q^2)}{q^9-q}.
\]
\end{theo}

\begin{proof}
Define the group action
\[
  \tau: \GL_3(q)\times\hat{V}_q\longrightarrow\hat{V}_q,\quad \big(S,(I,A_1,A_2)\big)\longmapsto (I,\,S^{-1}A_1S,\,S^{-1}A_2S).
\]
For every triple $(I,A_1,A_2)\in\hat{V}_q$, the orbit $\Orbt(I,A_1,A_2)$ contains a triple from $\hat{X}_q$ because~$A_1$ is similar to $C_f$, where $f=\chi_{A_1}$.
In other words, $\Orbt(I,A_1,A_2)=\Orbt(I,C_f,Z)$ for some $(I,C_f,Z)\in\hat{X}_q$.
This implies
\begin{equation}\label{e-Vunion}
  \hat{V}_q=\bigcup_{\hat{X}_q}\Orbt(I,C_f,Z).
\end{equation}
It remains to determine the sizes of the orbits $\Orbt(I,C_f,Z)$ for $(I,C_f,Z)\in\hat{X}_q$ as well as which of these orbits are distinct.
Not surprisingly, the orbits behave differently depending on whether or not $Z\in\subspace{I,C_f,C_f^2}$.
We thus define
\[
  \hat{X}_q^{(1)}=\{(I,C_f,Z)\in\hat{X}_q\mid Z\in\subspace{I,C_f,C_f^2}\},\quad \hat{X}_q^{(2)}=\hat{X}_q\setminus\hat{X}_q^{(1)}.
\]
Using the orbit-stabilizer theorem along with Proposition~\ref{P-StabCompMat} and Lemma~\ref{L-StabCfZ} we obtain
\begin{equation}\label{e-OrbSize}
  |\Orbt(I,C_f,Z)|=\left\{\begin{array}{cl}
      \frac{|\GL_3(q)|}{q^3-1},&\text{if }(I,C_f,Z)\in\hat{X}_q^{(1)},\\[1ex]
      \frac{|\GL_3(q)|}{q-1},&\text{if }(I,C_f,Z)\in\hat{X}_q^{(2)}.\end{array}\right.
\end{equation}
As for the number of distinct orbits, we clearly have $\Orbt(I,C_f,Z)\neq\Orbt(I,C_g,Z')$ whenever $f\neq g$.
Thus it remains to characterize when $\Orbt(I,C_f,Z)=\Orbt(I,C_f,Z')$.
To this end fix some $f\in\cI$ and consider the conjugation action
\[
  \phi:\Stab(C_f)\times\{Z\mid (I,C_f,Z)\in\hat{X}_q\}\longrightarrow \{Z\mid (I,C_f,Z)\in\hat{X}_q\},\quad
    (S,Z)\longmapsto S^{-1}ZS.
\]
Then Lemma~\ref{L-StabCfZ}  tells us that the stabilizers of this action are of the form
\[
   \Stab_\phi(Z)=
   \left\{\begin{array}{cl}\Stab(C_f),&\text{if }(I,C_f,Z)\in\hat{X}_q^{(1)},\\[.5ex]
   \{\alpha I\mid \alpha\in\F_q^*\},&\text{if }(I,C_f,Z)\in\hat{X}_q^{(2)}.
   \end{array}\right.
\]
Hence we conclude that
\begin{equation}\label{e-OrbZ1}
   \bigcup_{\;\hat{X}_q^{(1)}}\Orbt(I,C_f,Z)\text{ is a disjoint union,}
\end{equation}
whereas
\begin{equation}\label{e-OrbZ2}
   \bigcup_{\;\hat{X}_q^{(2)}}\Orbt(I,C_f,Z)\text{ is a union of $r$ distinct orbits, where }r=\frac{q-1}{|\Stab(C_f)|}|\hat{X}_q^{(2)}|.
\end{equation}
Now we can determine the cardinality of~$\hat{V}_q$.
Recalling that there are $(q^3-q)/3$ monic irreducible polynomials of degree~$3$ over~$\F_q$ (see \cite[Thm.~3.25]{LiNi97}) we have
\begin{equation}\label{e-Xsize}
  |X_q|=\frac{q^3-q}{3}(q^9-q^2)\ \text{ and }\ |\hat{X}_q^{(1)}|=\frac{q^3-q}{3}(q^3-q^2).
\end{equation}
The first identity follows from the fact that for the set~$X_q$ the matrix~$Z$ is any element in $\FFF\setminus\subspace{I,C_f}$, while for the set
$\hat{X}_q^{(1)}$ it is any element in $\subspace{I,C_f,C_f^2}\setminus\subspace{I,C_f}$.
Now~\eqref{e-OrbZ1},~\eqref{e-OrbZ2}, and~\eqref{e-OrbSize} along with $|\GL_3(q)|=\prod_{j=0}^2(q^3-q^j)$ lead to
\begin{align*}
   \Big|\!\!\bigcup_{\;\hat{X}_q^{(1)}}\Orbt(I,C_f,Z)\Big|&=\frac{q^3-q}{3}(q^3-q^2)\frac{|\GL_3(q)|}{q^3-1}=\frac{(q^3-q)^2(q^3-q^2)^2}{3},\\[1ex]
   \Big|\!\!\bigcup_{\;\hat{X}_q^{(2)}}\Orbt(I,C_f,Z)\Big|&=|\hat{X}_q^{(2)}|\frac{q-1}{q^3-1}\frac{|\GL_3(q)|}{q-1}=|\hat{X}_q^{(2)}|(q^3-q)(q^3-q^2).
\end{align*}
Notice further that $|\hat{X}_q^{(2)}|=|\hat{X}_q|-|\hat{X}_q^{(1)}|=|\hat{X}_q|-\frac{(q^3-q)(q^3-q^2)}{3}$.
Thus thanks to~\eqref{e-Vunion} we have
\[
  |\hat{V}_q|=\frac{(q^3-q)^2(q^3-q^2)^2}{3}+|\hat{X}_q^{(2)}|(q^3-q)(q^3-q^2)=|\hat{X}_q|(q^3-q)(q^3-q^2).
\]
Using the cardinalities~$|V_q|$ from~\eqref{e-WV} and $|X_q|$ from~\eqref{e-Xsize} we finally obtain the desired identity
\[
  \frac{|\hat{V}_q|}{|V_q|}=\frac{|\hat{X}_q|(q^3-q)(q^3-q^2)}{|X_q|\prod_{i=1}^2(q^9-q^i)}\frac{q^3-q}{3}(q^9-q^2)
  =\frac{|\hat{X}_q|}{|X_q|}\cdot\frac{1}{3}\frac{(q^3-q)^2(q^3-q^2)}{q^9-q}.
  \qedhere
\]
\end{proof}

Now we can prove Theorem~\ref{T-Reduction}.
The difference between the triples in the sets~$X_q$ and~$\cS$ is the shape of the first column of~$Z$, which we will now take into account.

\medskip

\noindent{\it Proof of Theorem~\ref{T-Reduction}.}
Consider the sets
\[
  Y_q=\{(I,C_f,Z)\in X_q\mid z_{11}=0=z_{21}\}\ \text{ and }\ \hat{Y}_q=Y_q\cap\hat{X}_q.
\]
Note that in the set $Y_q$, the matrix~$Z$ may have a zero entry at position $(3,1)$ as well, whereas this is not possible
for the set~$\hat{Y}_q$ thanks to the MRD property.
This also means that $\cS$ is obtained from $\hat{Y}_q$ by normalizing the entry at position~$(3,1)$.
Indeed, using the obvious fact that for any $\alpha\in\F_q^*$ the triple $(I,C_f,Z)$ generates an MRD code iff $(I,C_f,\alpha Z)$ generates an MRD code, we obtain
\begin{equation}\label{e-YS}
   |\hat{Y}_q|=(q-1)|\cS|.
\end{equation}
Furthermore, by the form of the first columns, every nonzero matrix~$Z$ with $z_{11}=z_{21}=0$ and every $f\in\cI$ gives rise to a linearly independent triple $(I,C_f,Z)$.
Using (again) that there are $(q^3-q)/3$ monic, irreducible, cubic polynomials over $\F_q$ and that~$Z$ has~$7$ free entries (not all of which are zero), we conclude
\begin{equation}\label{e-YCard}
  |Y_q|=\frac{q^3-q}{3}(q^7-1).
\end{equation}
It remains to relate the sets~$X_q$ and~$Y_q$.
This is easily achieved by the surjective map
\[
   \xi: X_q\longrightarrow Y_q,\quad (I,C_f,Z)\longmapsto (I,C_f,Z-z_{11}I-z_{21}C_f).
\]
Clearly, every fiber has cardinality $q^2$.
Since $\xi(\hat{X}_q)=\hat{Y}_q$ we obtain $|X_q|=q^2|Y_q|$ and $|\hat{X}_q|=q^2|\hat{Y}_q|$.
With the aid of~\eqref{e-YCard} and~\eqref{e-YS} we conclude
\[
   \frac{|\hat{X}_q|}{|X_q|}=\frac{|\hat{Y}_q|}{|Y_q|}=\frac{3(q-1)}{(q^3-q)(q^7-1)}|\cS|.
\]
Now the result follows from Proposition~\ref{P-BasisMatrix} and Theorem~\ref{T-MatComp}.
\hfill$\square$

\section{Parametrization of the set~\mbox{$\cS$}: Proof of Theorem~\ref{T-FinalCount}}\label{S-Thm2}

In this section we determine the precise cardinality of the set~$\cS$ defined in Theorem~\ref{T-Reduction}.
It will follow from a detailed study of the papers~\cite{Men73,Men77} by Menichetti on three-dimensional division algebras over finite fields.
While those papers~\cite{Men73,Men77} had the goal to count the isomorphism classes of these division algebras, we need to make sure that we count
\emph{all} division algebras, i.e. determine the cardinality of~$\cS$.
This makes it necessary to go into the details of the arguments in~\cite{Men73}.

We present the results in the matrix-theoretical form needed for determining~$|\cS|$ and without the notion of semifields.
In the next section we will elaborate on the connection to semifields and outline the ideas of~\cite{Men73,Men77}.
This will also provide us with a sketch of the proof of Theorem~\ref{T-MatSigma} below.
Recall the set~$\cS$ from Theorem~\ref{T-Reduction}, which in short we may write as
\begin{equation}\label{e-S2}
  \cS=\left\{\left({\small \begin{pmatrix}1&0&0\\0&1&0\\0&0&1\end{pmatrix},\ \begin{pmatrix}0&0&a\\1&0&b\\0&1&c\end{pmatrix},\
                         \begin{pmatrix}0&z_1&z'_1\\0&z_2&z'_2\\1&z_3&z'_3\end{pmatrix}}\right)\,\bigg|\, \text{MRD}\right\}
\end{equation}
This indicates that~$\cS$ may be of order~$q^9$ (in the worst case).
However, in order to establish our ultimate goal, $\lim_{q\rightarrow\infty}|\hat{T}_q|/|T_q|=0$, Theorem~\ref{T-Reduction} requires us to show that~$\cS$ is of
order at most~$q^8$.
In this section we will prove Theorem~\ref{T-FinalCount}, which shows that~$\cS$ is actually of order~$q^6$.

From now on we use the standard notation $k\bra{i}:=k^{q^i}$ for $k\in\F_{q^3}$. Thus $k=k\bra{0}$.
Recall that $k\mapsto k\bra{i}$ is an $\F_q$-linear automorphism of $\F_{q^3}$.
We need the following maps.

\begin{defi}\label{D-sigma}
Define the maps
\begin{align*}
  \sigma_1&:\F_{q^3}\longrightarrow\F_q,\quad k\longmapsto k\bra{0}+k\bra{1}+k\bra{2},\\[.5ex]
  \sigma_2&:\F_{q^3}\longrightarrow\F_q,\quad k\longmapsto k\bra{0}k\bra{1}+k\bra{0}k\bra{2}+k\bra{1}k\bra{2},\\[.5ex]
  \sigma_3&:\F_{q^3}\longrightarrow\F_q,\quad k\longmapsto k\bra{0}k\bra{1}k\bra{2}.
\end{align*}
\end{defi}

Obviously, $\sigma_1$ is the trace map of the field extension $\F_{q^3}\,|\,\F_q$ and $\sigma_3$ the norm.
More generally, $\sigma_1,\sigma_2,\sigma_3$ are the symmetric functions evaluated at $k\bra{0},k\bra{1},k\bra{2}$ and thus
\begin{equation}\label{e-sigmak}
   (x-k\bra{0})(x-k\bra{1})(x-k\bra{2})=x^3-\sigma_1(k)x^2+\sigma_2(k)x-\sigma_3(k)\ \text{ for all }\ k\in\F_{q^3}.
\end{equation}
This also implies for all $k,\ell\in\F_{q^3}$
\begin{equation}\label{e-sigmaInv}
    \sigma_j(\ell)=\sigma_j(k)\text{ for }j=1,2,3\Longleftrightarrow \ell\in\{k\bra{0},k\bra{1},k\bra{2}\}.
\end{equation}

Furthermore, we need the map
\begin{equation}\label{e-phi}
   \phi:\F_{q^3}^2\longrightarrow\F_{q^3},\quad (k,\hat{k})\longmapsto (k+\hat{k})\big(\sigma_1(\hat{k})-\hat{k}\big)-\sigma_2(\hat{k}).
\end{equation}
For further use we record
\begin{equation}\label{e-phikk}
   \phi(k,k\bra{1})=\phi(k,k\bra{2})=k^2.
\end{equation}

Now we are ready to present the crucial step in our count.
The following result has been established by Menichetti~\cite{Men73} and amounts to a parametrization of the set~$\cS$.
In the next section we will outline~\cite{Men73} and sketch a proof of this result.
This will also provide some insight into the special cases $\hat{k}=k$ and $\hat{k}\in\{k\bra{1},k\bra{2}\}$.
For instance, Proposition~\ref{P-Fcommass2} will show that $\hat{k}=k$ leads to the only non-associative commutative semifield for the given~$k$
(provided that $1,k,\phi(k,k)$ are linearly independent).

\begin{theo}[\mbox{\cite[Prop.~10]{Men73}}]\label{T-MatSigma}
For $(k,\hat{k})\in\F_{q^3}^2$ define
\[
  \Sigma_1(k,\hat{k}):=\begin{pmatrix}0\!&\!0\!&\!\sigma_3(k)\\1\!&\!0\!&\!-\sigma_2(k)\\0\!&\!1\!&\!\sigma_1(k)\end{pmatrix},\
  \Sigma_2(k,\hat{k}):=
        \begin{pmatrix}0\!&\!\sigma_3(\hat{k})\!&\!\sigma_1(k)\sigma_3(\hat{k})+\sigma_1(\hat{k})\sigma_3(k)-\sigma_2(k\hat{k})\\
                               0\!&\!-\sigma_2(\hat{k})\!&\!-\sigma_3(k+\hat{k})\\
                               1\!&\!\sigma_1(\hat{k})\!&\!\sigma_1(k)\sigma_1(\hat{k})-\sigma_1(k\hat{k})\end{pmatrix}.
\]
Then the set~$\cS$ from~\eqref{e-S2} is given by
\[
   \cS=\big\{\big(I,\Sigma_1(k,\hat{k}),\Sigma_2(k,\hat{k})\big)\,\big|\, k,\hat{k}\in\F_{q^3}
   \text{ such that  $1,\, k,\, \phi(k,\hat{k})$ are linearly independent over }\F_q\big\}.
\]
\end{theo}

Let us briefly comment on this result.
Given $(I,C_f,Z)\in\cS$ as in~\eqref{e-S2}.
Then it is clear from~\eqref{e-sigmak} that for the identity $(C_f,Z)=\big(\Sigma_1(k,\hat{k}),\Sigma_2(k,\hat{k})\big)$ to be true,
the parameters~$k$ and~$\hat{k}$ must be roots of the irreducible polynomials~$x^3-cx^2-bx-a$ and
$x^3-z_{3}x^2-z_{2}x-z_{1}\in\F_q[x]$ in $\F_{q^3}$, respectively.
The main part of the proof consists of showing that the last columns of $\Sigma_2(k,\hat{k})$ and~$Z$ coincide, and that the MRD property translates into
the linear independence of $1,\, k,\, \phi(k,\hat{k})$.
Details will be given in the next section.

Note that the result above shows that~$|\cS|$ is of order at most~$q^6$.
In order to determine the exact cardinality, we need to investigate when two pairs $(k,\hat{k})$ give rise to the same matrix pair
$(\Sigma_1(k,\hat{k}),\Sigma_2(k,\hat{k}))$.
One easily verifies that the linear independence of $1,k,\phi(k,\hat{k})$ implies that $k,\hat{k}\in\F_{q^3}\setminus\F_q$, and therefore one may restrict oneself to that situation.
This is dealt with in the next proposition, part of which has been proven in \cite[Lem.~12.1]{Men73}.
For sake of completeness we include a proof of all statements below.

\begin{prop}\label{P-SigmaInj}
Let $(k,\hat{k}),\,(\ell,\hat{\ell})\in(\F_{q^3}\setminus\F_q)^2$. Then
\[
   (\ell,\,\hat{\ell})=(k\bra{r},\,\hat{k}\bra{r})\text{ for some }r\in\{0,1,2\}\Longrightarrow
   \big(\Sigma_1(k,\hat{k}),\,\Sigma_2(k,\hat{k})\big)=\big(\Sigma_1(\ell,\hat{\ell}),\,\Sigma_2(\ell,\hat{\ell})\big).
\]
Furthermore,
\begin{arabiclist}
\item Let $\hat{k}\not\in\{k\bra{1},k\bra{2}\}$. Then
        \[
          \big(\Sigma_1(k,\hat{k}),\,\Sigma_2(k,\hat{k})\big)=\big(\Sigma_1(\ell,\hat{\ell}),\,\Sigma_2(\ell,\hat{\ell})\big)
          \Longleftrightarrow
          \exists\;r\in\{0,1,2\} \text{ such that }(\ell,\,\hat{\ell})=(k\bra{r},\,\hat{k}\bra{r}).
        \]
\item Let $\hat{k}=k\bra{n}$ for some $n\in\{1,2\}$. Then
        \[
          \big(\Sigma_1(k,\hat{k}),\,\Sigma_2(k,\hat{k})\big)=\big(\Sigma_1(\ell,\hat{\ell}),\,\Sigma_2(\ell,\hat{\ell})\big)
          \!\Longleftrightarrow\!
          \left\{\!\!\begin{array}{l}
          \exists\;r\in\{0,1,2\} \text{ such that}\\
          (\ell,\,\hat{\ell})=(k\bra{r},\,k\bra{n+r})\text{ or }(\ell,\,\hat{\ell})=(k\bra{n+r},\,k\bra{r}).
          \end{array}\right.
        \]
\end{arabiclist}
\end{prop}

\begin{proof}
The first implication follows immediately from~\eqref{e-sigmaInv} and the definition of the matrices~$\Sigma_1,\Sigma_2$.
The additional backwards implication in~(2), where $\hat{k}=k\bra{n}$ and $(\ell,\,\hat{\ell})=(k\bra{n+r},\,k\bra{r})$ follows again from~\eqref{e-sigmaInv}  together with
the fact that in this case $\ell+\hat{\ell}=k\bra{n+r}+k\bra{r}=(\hat{k}+k)\bra{r}$ and $\ell\hat{\ell}=k\bra{n+r}k\bra{r}=(\hat{k}k)\bra{r}$.

Let us now assume $\big(\Sigma_1(k,\hat{k}),\,\Sigma_2(k,\hat{k})\big)=\big(\Sigma_1(\ell,\hat{\ell}),\Sigma_2(\ell,\hat{\ell})\big)$.
Then
\[
    \sigma_j(\ell)=\sigma_j(k),\quad \sigma_j(\hat{\ell})=\sigma_j(\hat{k}),\quad \sigma_j(\ell\hat{\ell})=\sigma_j(k\hat{k}),\quad
    \sigma_j(\ell+\hat{\ell})=\sigma_j(k+\hat{k})\text{ for }j=1,2,3.
\]
Indeed: The first two identities are clear from the matrices; for $j=1,2$ the identity for the product also follows from comparing the matrices, while for $j=3$ it is a consequence of
the multiplicativity of~$\sigma_3$;
finally, for $j=3$ the identity for the sum follows again from the matrices, for $j=1$ from the additivity of $\sigma_1$, and for $j=2$ from the identity
$\sigma_2(x+y)=\sigma_2(x)+\sigma_2(y)-\sigma_1(xy)+\sigma_1(x)\sigma_1(y)$, see~\cite[Lem.~10.3]{Men73}.

With the aid of~\eqref{e-sigmaInv} all of this tells us that
\begin{equation}\label{e-alldata}
  \ell=k\bra{i},\ \hat{\ell}=\hat{k}\bra{j},\ \ell+\hat{\ell}=(k+\hat{k})\bra{r},\ \ell\hat{\ell}=(k\hat{k})\bra{s}\text{ for some }i,j,r,s\in\{0,1,2\}.
\end{equation}
We have to show that $i=j=r=s$. Let us assume $i\neq j$.
Then $r\not\in\{i,j\}$ by the third identity in~\eqref{e-alldata} and $s\not\in\{i,j\}$ by the fourth identity.
Since $i,j,r,s$ attain only $3$ possible values, we conclude $r=s$.
Thus
\[
  k\bra{i}+\hat{k}\bra{j}=k\bra{r}+\hat{k}\bra{r}\ \text{ and }\ k\bra{i}\hat{k}\bra{j}=k\bra{r}\hat{k}\bra{r}.
\]
This implies $(x-k\bra{i})(x-\hat{k}\bra{j})=(x-k\bra{r})(x-\hat{k}\bra{r})$ and therefore $k\bra{i},\hat{k}\bra{j}\in\{k\bra{r},\hat{k}\bra{r}\}$.
Since $r\not\in\{i,j\}$ we thus have
\begin{equation}\label{e-khatk}
    k\bra{i}=\hat{k}\bra{r}\ \text{ and }\ \hat{k}\bra{j}=k\bra{r}.
\end{equation}
This in turn implies $\hat{k}=k\bra{i+3-r}$. Now we may argue as follows.
\begin{romanlist}
\item If $\hat{k}\not\in\{k,k\bra{1},k\bra{2}\}$, then the last identity is a contradiction.
         Hence in this case we arrive at $i=j=r=s$.
\item If $\hat{k}=k$ then~\eqref{e-khatk} leads to the contradiction $i=r$. Thus again $i=j=r=s$ by virtue of~\eqref{e-alldata}.
        This and~(i) establishes Part~(1) of the proposition.
\item Consider now Part~(2), thus $\hat{k}=k\bra{n}$ for some $n\in\{1,2\}$.
The case $i=j$ leads to the first option in the implication.
If $i\neq j$, then~\eqref{e-khatk} yields $i\equiv n+r~\mod 3$ and $j+n\equiv r~\mod 3$.
This means $(\ell,\hat{\ell})=(k\bra{n+r},k\bra{r})$, which is the second option. \qedhere
\end{romanlist}
\end{proof}

\begin{cor}\label{C-SizeS}
The set~$\cS$ has cardinality $|\cS|=\frac{1}{3}(|\cS'|+\frac{1}{2}|\cS''|)$,
where
\begin{align*}
   \cS'&=\big\{(k,\hat{k})\mid  1,\, k,\, \phi(k,\hat{k}) \text{ are linearly independent over $\F_q$ and }\hat{k}\not\in\{k\bra{1},k\bra{2}\}\big\},\\
   \cS''&=\big\{(k,k\bra{n})\mid k\in\F_{q^3}\setminus\F_q,\,n\in\{1,2\}\big\}.
\end{align*}
\end{cor}

\begin{proof}
The summand $1/3|\cS'|$ follows from Proposition~\ref{P-SigmaInj}(1) along with the simple fact that, for any~$r$, the linear independence of $1,k,\phi(k,\hat{k})$
implies the linear independence of $1,k\bra{r},\phi(k\bra{r},\hat{k}\bra{r})$.
As for the second summand, note first that~\eqref{e-phikk} guarantees that $1,k,\phi(k,k\bra{n})$ are linearly independent for any $k\in\F_{q^3}\setminus\F_q$ and $n\in\{1,2\}$.
Thus Proposition~\ref{P-SigmaInj}(2) leads to the summand $1/6|\cS''|$.
\end{proof}

The following result appears implicitly in~\cite{Men77}.

\begin{prop}\label{P-CardSet}
Consider the set~$\cS'$ from Corollary~\ref{C-SizeS}. Then
$|\cS'|=(q^3-q)(q^3-q^2-q-2)$.
In particular $\cS'=\emptyset$ iff $q=2$.
\end{prop}

\begin{proof}
We follow  \cite[pp.~404]{Men77}.
If $1,\, k,\, \phi(k,\hat{k})$ are linearly independent, then $k\in\F_{q^3}\setminus\F_q$.
Fix any such~$k$.
Then $(1,\,k,\,k^2)$ forms a basis of the $\F_q$-vector space $\F_{q^3}$.
Write $\hat{k}=x_0+x_1k+x_2k^2,\,x_i\in\F_q$.
We want to determine the number of elements~$\hat{k}$ such that $1,k,\phi(k,\hat{k})$ are linearly independent.
By \cite[Cor.~10.2]{Men73} the linear independence of $1,k,\phi(k,\hat{k})$ is equivalent to the linear independence of $1,\hat{k},\phi(\hat{k},k)$.
We determine the coordinate vectors of $1,\hat{k},\phi(\hat{k},k)$ w.r.t.\ the basis $(1,\,k,\,k^2)$.
Let $f=x^3-a_2x^2-a_1x-a_0\in\F_q[x]$ be the minimal polynomial of~$k$.
Thus $a_0=\sigma_3(k),\, a_1=-\sigma_2(k)$, and $a_2=\sigma_1(k)$.
Then
\[
   \phi(\hat{k},k)=(\hat{k}+k)(a_2-k)+a_1=a_2x_0+a_1-a_0x_2+(a_2(1+x_1)-x_0-a_1x_2)k-(1+x_1)k^2.
\]
Hence the coordinate vectors of $1,\,\hat{k},\,\phi(\hat{k},k)$ are the columns of the matrix
\[
   M(x_0,x_1,x_2)=\begin{pmatrix}1&x_0&a_2x_0+a_1-a_0x_2\\0&x_1&a_2(1+x_1)-x_0-a_1x_2\\0&x_2&-(1+x_1)\end{pmatrix},
\]
and we need to determine the number of points $(x_0,x_1,x_2)\in\F_q^3$ such that $\det M(x_0,x_1,x_2)\not=0$.
It is easy to verify that the curve given by
\[
    \det M(x_0,x_1,x_2)=(x_0+a_1x_2)x_2-(x_1+1)(x_1+a_2x_2)=0
\]
has exactly $q^2+q$ solutions
($2q$ solutions for $x_2=0$ and $(q-1)q$ solutions for $x_2\neq0$).
All of this tells us that for every $k\in\F_{q^3}\setminus\F_q$ there exist $q^3-q^2-q$ elements~$\hat{k}$ such that
$1,\,k,\,\phi(k,\hat{k})$ are linearly independent.
Finally, if $\hat{k}=k\bra{n}$ for some $n\in\{1,2\}$, then $1,\,k,\,\phi(k,\hat{k})$ are linearly independent by~\eqref{e-phikk}.
Thus we have to subtract these two elements from our count.
This yields the desired result.
\end{proof}

Now the proof of Theorem~\ref{T-FinalCount} follows quickly.

\medskip
\noindent{\it Proof of Theorem~\ref{T-FinalCount}.}
Since $|\cS''|=2(q^3-q)$, the stated cardinality of the set~$\cS$ follows from Proposition~\ref{P-CardSet} and Corollary~\ref{C-SizeS}.
The proportion of $[3\times3;3]$-MRD codes is now immediate with Theorem~\ref{T-Reduction}.
\hfill$\square$

\section{Relation to $3$-Dimensional Semifields}\label{S-Semifields}
In this section we recount the relation between MRD codes and semifields and go on to survey the paper~\cite{Men73} by Menichetti as needed for the proof of Theorem~\ref{T-MatSigma}.
This will also provide further insight into the various types of semifields parametrized by $(k,\hat{k})$ as in Theorem~\ref{T-MatSigma}.

Menichetti used this theorem to determine the number of isomorphism classes of the related non-associative semifields, which are in this case $3$-dimensional over~$\F_q$.
It turns out~\cite[p.~400]{Men77} that for $q\geq3$ the number of isomorphism classes of non-associative semifields is
\[
   \frac{q^3-q^2+q-10}{3}\text{ if }q\equiv1~\mod~3 \quad \text{ and }\quad \frac{q^3-q^2+q-6}{3}\text{ if }q\not\equiv1~\mod~3.
\]
Since these cardinalities agree with the number of isomorphism classes of $3$-dimensional twisted fields, determined by Kaplansky~\cite[p.~77]{Ka75}, Menichetti thereby proved that all
$3$-dimensional non-associative semifields are twisted fields.

We start by recalling the notion of a semifield.

\begin{defi}\label{D-SF}
A \emph{finite semifield} is a triple $(\cF,+,\circ)$ with a finite set~$\cF$ and binary operations~$+$ and~$\circ$ on~$\cF$ satisfying the following conditions.
\begin{romanlist}
\item $(\cF,+)$ is an abelian group,
\item $\cF$ has no left or right zero divisors, that is: $a\circ b=0\Longrightarrow a=0\text{ or }b=0$.
\item Both distributivity laws hold: $a\circ(b+c)=a\circ b+a\circ c$ and $(a+b)\circ c=a\circ c+b\circ c$.
\item There exists an identity element: $1\circ a=a\circ 1=a$ for all $a\in\cF$.
\end{romanlist}
A semifield without identity element is called a \emph{pre-semifield}.
\end{defi}

It is easy to see that every finite semifield contains a field, say~$\F_q$, and is an algebra over~$\F_q$.
Thus  $(\lambda a)\circ b=\lambda(a\circ b)=a\circ(\lambda b)$ for all $a,b\in\cF$ and $\lambda\in\F_q$.
If~$\cF$ is $n$-dimensional over~$\F_q$, we call~$\cF$ an \emph{$n$-dimensional division algebra}.

Note that a (pre-) semifield is not necessarily associative or commutative.
Clearly, an associative finite semifield is a finite skew field and thus a finite field by Wedderburn's Theorem, hence commutative.
Later in Proposition~\ref{P-Fcommass2} we will encounter instances of commutative non-associative semifields.

It is well-known from finite geometry that $[n\times n;n]_q$-MRD codes are in one-one-correspondence to $n$-dimensional division algebras over~$\F_q$.
For sake of completeness and further use we provide the correspondence below.
Most properties are easy to verify.
For details see also~\cite[p.~220]{Dem68} or~\cite[p.~134]{LaPo11} or other sources on finite geometry.

\begin{prop}\label{P-SFMRD}
\begin{alphalist}
\item Let $A_1,\ldots,A_n\in\F_q^{n\times n}$ generate an $[n\times n;n]$-MRD code.
        Define the linear map
        $\F_q^n\longrightarrow \subspace{A_1,\ldots,A_n},\ x\longmapsto M_x$,
        where $M_x$ is the unique matrix in $\subspace{A_1,\ldots,A_n}$ whose first column is~$x$.
        On~$\F_q^n$ we define
        \[
             x\circ y:=M_x y,
        \]
        where the right hand side is ordinary matrix-vector multiplication.
        Then $(\F_q^n,+,\circ)$ is a pre-semifield.
        It is a semifield, and hence an $n$-dimensional division algebra, iff the code $\subspace{A_1,\ldots,A_n}$ contains the identity matrix.
\item Let~$\cF$ be an $n$-dimensional division algebra over~$\F_q$, and $u_1,\ldots,u_n$ be a basis of~$\cF$ over~$\F_q$.
        For every $a\in\cF$ define the (\,$\F_q$-linear) map $m_a:\cF\longrightarrow\cF,\ y\longmapsto a\circ y$.
        Let~$M_h$ be the matrix representation of $m_{u_h}$ w.r.t.\ to the basis $u_1,\ldots,u_n$.
        Then the subspace $\subspace{M_1,\ldots,M_n}$ is an $[n\times n;n]$-MRD code containing the identity matrix.
        The matrices $M_1,\ldots,M_n$ are the \emph{structure matrices} of~$\cF$.
\end{alphalist}
\end{prop}

Note that the matrices~$M_x$ in~(a) are well-defined.
Indeed, since every nonzero linear combination has rank~$n$, the first columns of the matrices in $\subspace{A_1,\ldots,A_n}$ are distinct and thus
assume each vector of~$\F_q^n$ exactly once.
More generally, the matrix~$M_x$  could be replaced by a matrix $\phi(x)$, where~$\phi$
is any isomorphism between~$\F_q^n$ and $\subspace{A_1,\ldots,A_n}$. In that case multiplication would read as
$x\circ y=\phi(x)y$.

The above relationship becomes particularly nice for the MRD codes generated by the matrix triples in the set~$\cS$ from Theorem~\ref{T-Reduction}.
Denote by $e_1,e_2,e_3\in\F_q^3$ the standard basis vectors.

\begin{defi}\label{D-SF3}
Let $(I,C_f,Z)\in\cS$ and
\begin{equation}\label{e-CfZ}
    C_f=\begin{pmatrix}0&0&a\\1&0&b\\0&1&c\end{pmatrix},\
    Z=\begin{pmatrix}0&z_1&z'_1\\0&z_2&z'_2\\1&z_3&z'_3\end{pmatrix}.
\end{equation}
We denote the associated $3$-dimensional division algebra by $\cF(I,C_f,Z)$.
For $x=(x_1,x_2,x_3)\in\F_q^3$ the matrix $M_x$ from Proposition~\ref{P-SFMRD}(a) is given by $M_x=x_1I+x_2C_f+x_3Z$.
Furthermore, $e_1$ is the identity of~$\cF$ and
\[
  e_2\circ e_2=e_3,\quad e_2\circ e_3=\begin{pmatrix}a\\b\\c\end{pmatrix},\quad
  e_3\circ e_2=\begin{pmatrix}z_1\\z_2\\z_3\end{pmatrix},\quad
  e_3\circ e_3=\begin{pmatrix}z'_1\\z'_2\\z'_3\end{pmatrix}.
\]
\end{defi}

Note in particular $(e_2\circ e_2)\circ e_2=(z_1,z_2,z_3)\T$, while $e_2\circ(e_2\circ e_2)=(a,b,c)\T$,
hence we have no well-defined notion of third (or higher) powers unless the semifield is associative.
The latter as well as commutativity are easily characterized.

\begin{prop}\label{P-Fcommass}
Let $(I,C_f,Z)\in\cS$ and $C_f,Z$ be as in~\eqref{e-CfZ}.
\begin{alphalist}
\item $\cF(I,C_f,Z)$ is commutative iff $(a,b,c)=(z_1,z_2,z_3)$.
\item $\cF(I,C_f,Z)$ is associative iff $Z=C_f^2$, in which case $\cF(I,C_f,Z)$ is the field $\F_{q^3}$.
\end{alphalist}
\end{prop}

\begin{proof}
(a) is immediate because the products between the standard basis vectors fully determine the multiplication in~$\cF$.
\\
(b) We compute
\[
  e_2\circ(e_2\circ e_3)=\!\begin{pmatrix}ac\\a\!+\!bc\\b\!+\!c^2\end{pmatrix}\!,\:
  (e_2\circ e_2)\circ e_3=\!\begin{pmatrix}z'_1\\z'_2\\z'_3\end{pmatrix}\!,\:
  e_2\circ(e_3\circ e_2)=\!\begin{pmatrix}az_3\\z_1\!+\!bz_3\\z_2\!+\!cz_3\end{pmatrix}\!,\:
  (e_2\circ e_3)\circ e_2=\!\begin{pmatrix}cz_1\\a\!+\!cz_2\\b\!+\!cz_3\end{pmatrix}.
\]
Suppose that $\cF(I,C_f,Z)$ is associative.
Comparing the last two expressions we obtain $z_2=b$ and
\[
  \begin{pmatrix}c&-a\\1&b\end{pmatrix}\begin{pmatrix}z_1\\z_3\end{pmatrix}=\begin{pmatrix}0\\a+bc\end{pmatrix}.
\]
Note that $a+bc=-\det(cI-C_f)\neq0$.
Hence $(z_1,z_3)=(a,c)$ and further
\[
    Z=\begin{pmatrix}0&a&ac\\0&b&a+bc\\1&c&b+c^2\end{pmatrix}=C_f^2.
\]
Conversely, if $Z=C_f^2$, then $\subspace{I,C_f,C_f^2}$ is the field~$\F_{q^3}$, and thus closed under multiplication.
As a consequence, the products $M_xM_y$ are in $\subspace{I,C_f,C_f^2}$ for any $x,y\in\F_q^3$.
Since the first column of $M_xM_y$ is given by $x\circ y$, we conclude that
the map $(\F_q^3,+,\circ)\longrightarrow\subspace{I,C_f,C_f^2},\ x\longmapsto M_x$ is multiplicative, i.e., an isomorphism of semifields, and thus $\cF(I,C_f,C_f^2)$ is associative.
\end{proof}

In the relation between semifields and MRD codes we made use of the multiplication maps~$m_x$, based on the fixed left factor~$x$.
It will be crucial to also consider the multiplication maps with given right factors.
This leads to a dual matrix triple (not to be confused with the trace-dual of MRD codes).

\begin{prop}[see also \mbox{\cite[Prop.~5]{Men73}}]\label{P-Opp}
Consider the division algebra $\cF(I,C_f,Z)$ from Definition~\ref{D-SF3} and the $\F_q$-linear maps
\[
  m_x:\cF\longrightarrow\cF,\ y\longmapsto x\circ y,\qquad
  \hat{m}_x:\cF\longrightarrow\cF,\ y\longmapsto y\circ x.
\]
Then the matrix representations of $m_{e_2}$ and $m_{e_3}$ w.r.t.\ the standard basis are $C_f$ and~$Z$, respectively, and the matrix representations of
$\hat{m}_{e_2}$ and $\hat{m}_{e_3}$ are
\begin{equation}\label{e-DualTriple}
   C_g:=\begin{pmatrix}0&0&z_1\\1&0&z_2\\0&1&z_3\end{pmatrix}\ \text{ and }\
   \hat{Z}=\begin{pmatrix}0&a&z'_1\\0&b&z'_2\\1&c&z'_3\end{pmatrix},
\end{equation}
respectively (using column vector notation throughout).
As a consequence, $\subspace{I,C_g,\hat{Z}}$ is an MRD code as well.
Thus $(I,C_g,\hat{Z})\in\cS$ and $g=x^3-z_3x^2-z_2x-z_1\in\F_q[x]$ is irreducible.
We call $(I,C_g,\hat{Z})$ the \emph{dual triple} to $(I,C_f,Z)$.
\end{prop}

\begin{proof}
The matrix representations of $m_{e_2}$ and $m_{e_3}$ are clear from the definition of~$\circ$.
For the matrix representation of $\hat{m}_{e_2}$ notice that its columns are the coordinate vectors of $e_1\circ e_2,\,e_2\circ e_2$, and $e_3\circ e_2$.
This results in the matrix~$C_g$. Similarly one verifies~$\hat{Z}$.
The MRD property of $\subspace{I,C_g,\hat{Z}}$ follows from the void of zero divisors in $\cF(I,C_f,Z)$.
The latter then also implies that~$g$ has no roots in~$\F_q$, thus is irreducible.
\end{proof}

Note that, by definition, $\hat{m}_{e_2}(y)=y\circ e_2=C_gy$ and $\hat{m}_{e_3}(y)=y\circ e_3=\hat{Z}y$ and all $y\in\F_q^3$.
Strictly speaking, the dual triple $(I,C_g,\hat{Z})$ are the structure matrices of the \emph{opposite semifield} defined by $x\circ' y:=y\circ x$.
This leads to $e_2\circ' y=C_gy$ and $e_3\circ' y=\hat{Z}y$  and avoids changing the order of the factors.

For the rest of this section we will also have to consider values for $x=(x_1,x_2,x_3)$ in the extension field $\F_{q^3}$.
We denote by $\P^2(q^3)$ and $\P^2(q)$ the projective planes over~$\F_{q^3}$ and $\F_q$, respectively.
Then $\P^2(q)\subseteq\P^2(q^3)$, and we call the points in $\P^2(q)$ the $\F_q$-\emph{rational points}.

One easily verifies~\cite[Eq.~(3)]{Men73} that a triple $(I,C_f,Z)\in\cS$ and its dual triple $(I,C_g,\hat{Z})$ satisfy
\begin{equation}\label{e-StructMat}
  x_1I+x_2C_f+x_3Z=(x\,|\,C_gx\,|\,\hat{Z}x) \text{  for any }x=(x_1,x_2,x_3)\T\in\F_{q^3}^3.
\end{equation}

We are now ready to give an outline of the proof of Theorem~\ref{T-MatSigma}.
For the next results we fix the following setting.

\begin{nota}\label{N-Notation}
We fix a triple $(I,C_f,Z)\in\cS$ with dual triple $(I,C_g,\hat{Z})\in\cS$ and where
$C_f,\,Z,\,C_g,\,\hat{Z}$ are as in~\eqref{e-CfZ} and~\eqref{e-DualTriple}.
\begin{romanlist}
\item Denote by
        $\alpha\bra{i},\,\beta\bra{i},\hat{\alpha}\bra{i}\in\F_{q^3},\,i=0,1,2,$ the eigenvalues of~$C_f,\,Z,\,C_g$, respectively.
        The eigenvalues do indeed appear in conjugate triples because the characteristic polynomial of every nonzero matrix in a $[3\times3;3]$-MRD code is irreducible.
\item Let $v=(v_1,v_2,v_3)\in\F_{q^3}^3$ be an eigenvector of~$C_g$ w.r.t.\ the eigenvalue~$\hat{\alpha}\bra{0}$.
        Then $v\bra{i}:=(v_1^{[i]},v_2^{[i]},v_3^{[i]})$ is an eigenvector of~$C_g$ w.r.t.~$\hat{\alpha}\bra{i}$ for $i=0,1,2$.
\item Let~$\Gamma$ be the curve in $\P^2(q^3)$ given by the equation $F(x_1,x_2,x_3)=0$, where
        \[
          F(x_1,x_2,x_3)=\det(x_1I+x_2 C_f+x_3 Z)\in\F_q[x_1,x_2,x_3].
        \]
\end{romanlist}
\end{nota}

We have the following simple properties.

\begin{prop}[\mbox{\cite[Props.~3, 6, 7]{Men73}}]\label{P-Facts}
Consider the setting from Notation~\ref{N-Notation}.
\begin{alphalist}
\item The curve~$\Gamma$ has no $\F_q$-rational points.
\item The points $w\bra{i}:=(-\alpha\bra{i},1,0)$ and $\tilde{w}\bra{i}=(-\beta\bra{i},0,1),\,i=0,1,2,$ are the intersection points of~$\Gamma$ with the
         curve given by $x_3=0$ and $x_2=0$, respectively.
\item The points $v\bra{i},\,i=0,1,2,$ are on~$\Gamma$. Furthermore, $v_2\neq0\neq v_3$.
\end{alphalist}
\end{prop}

\begin{proof}
(a) is clear because every nontrivial linear combination $x_1I+x_2 C_f+x_3 Z$ with coefficients~$x_i$ in~$\F_q$ is invertible thanks to the MRD property.
\\
(b) is obvious and the first part of~(c) follows from~\eqref{e-StructMat}.
For the second statement in~(c) consider $C_g(v_1,v_2,v_3)\T=(z_1v_3,\,v_1+z_2v_3,\,v_2+z_3v_3)\T=\hat{\alpha}(v_1,v_2,v_3)\T$.
If $v_3=0$, then  also $v_1=v_2=0$, a contradiction.
If $v_2=0\neq v_3$, then $z_3=\hat{\alpha}$, which is a contradiction because $z_3$ is in~$\F_q$ whereas~$\hat{\alpha}$ is not.
Thus $v_2\neq0\neq v_3$.
\end{proof}

The next result is a crucial step for the proof of Theorem~\ref{T-MatSigma} and involves a geometric argument.

\begin{prop}[\mbox{\cite[Lem.~9.1, Prop.~9]{Men73}}]\label{P-ReducGamma}
Consider the setting from Notation~\ref{N-Notation}.
The curve~$\Gamma$ is the union of three conjugate lines (with coefficients in~$\F_{q^3}$).
More precisely,
\[
   F(x_1,x_2,x_3)=\prod_{i=0}^2(x_1+\alpha\bra{i}x_2+\beta\bra{i}x_3).
\]
\end{prop}

\noindent{\it Sketch of Proof.} (See~\cite{Men73})
The polynomial $F(x_1,x_2,x_3)=\det(x_1I+x_2 C_f+x_3 Z)$ has coefficients in~$\F_q$.
If~$F$ was irreducible in $\F_{q^3}[x_1,x_2,x_3]$, then the Hasse-Weil inequality $|N-(q+1)|\leq 2\gamma\sqrt{q}$,
where~$N$ is the number of~$\F_q$-rational points and~$\gamma$ the genus of~$\Gamma$, would imply $N>0$, contradicting Proposition~\ref{P-Facts}(a).
Hence~$F$ is reducible in $\F_{q^3}[x_1,x_2,x_3]$.
Suppose~$F$ has a quadratic irreducible factor, and thus also a linear factor.
Since~$F$ is invariant under the Frobenius homomorphism, extended coefficientwise to $\F_{q^3}[x_1,x_2,x_3]$, the same is true for these irreducible factors (up to a constant factor).
Hence the factors are in $\F_q[x_1,x_2,x_3]$, in contradiction to Proposition~\ref{P-Facts}(a).
Thus~$F$ factors into~$3$ (mutually conjugate) lines.
Now we consider the~$6$ points on~$\Gamma$ given in Proposition~\ref{P-Facts}(b) and argue as follows.
\begin{romanlist}
\item $w\bra{0},w\bra{1},w\bra{2}$ are on the line $x_3=0$, and this line is clearly not a factor of~$F(x_1,x_2,x_3)$ thanks to Proposition~\ref{P-Facts}(a).
\item $\tilde{w}\bra{0},\tilde{w}\bra{1},\tilde{w}\bra{2}$ are on the line $x_2=0$, which is also not a factor.
\item Since~$\Gamma$ is the union of three lines, each~$w\bra{i}$ must be on the same line as a point $\tilde{w}\bra{j}$.
        After appropriate re-indexing this results in the three lines given in the statement.
        \hfill$\square$
\end{romanlist}

\medskip

\noindent{\it Proof of Theorem~\ref{T-MatSigma}.} (See \cite[pp.~291--293]{Men73})
We have to establish the stated identity for the set~$\cS$.
\\
``$\subseteq$''
Let $(I,C_f,Z)\in\cS$ with dual triple $(I,C_g,\hat{Z})$ and where $C_f,Z,C_g,\hat{Z}$ are as in~\eqref{e-CfZ} and~\eqref{e-DualTriple}.
Recall the map~$\phi$ from~\eqref{e-phi}.
We have to show that there exist $k,\hat{k}\in\F_{q^3}$ such that $1,k,\phi(k,\hat{k})$ are linearly independent and $(C_f,Z)=(\Sigma_1(k,\hat{k}),\Sigma_2(k,\hat{k}))$.
Recall that~$C_f$ and $C_g$ are the companion matrices of the irreducible polynomials $f=x^3-cx^2-bx-a$ and $g=x^3-z_3x^2-z_2x-z_1$, respectively.
Let $k\bra{i}$ and $\hat{k}\bra{i},\,i=0,1,2,$ be the roots of~$f$ and~$g$ in $\F_{q^3}$, respectively.
Furthermore, let $\beta\bra{i}\in\F_{q^3},\,i=0,1,2,$ be the eigenvalues of~$Z$.
Now we are in the setting of Notation~\ref{N-Notation} with $\alpha\bra{i}=k\bra{i}$ and $\hat{\alpha}\bra{i}=\hat{k}\bra{i}$.
Thanks to~\eqref{e-sigmak} we have
\begin{equation}\label{e-abczsigma}
  a=\sigma_3(k),\ b=-\sigma_2(k),\ c=\sigma_1(k),\ z_1=\sigma_3(\hat{k}),\ z_2=-\sigma_2(\hat{k}),\ z_3=\sigma_1(\hat{k}).
\end{equation}
In particular $C_f=\Sigma_1(k,\hat{k})$.
Since $\hat{k}\bra{i}$ are the eigenvalues of~$C_g$, the matrices
\[
   C_g-\hat{k}\bra{i}I=\begin{pmatrix}-\hat{k}\bra{i}&0&z_1\\1&-\hat{k}\bra{i}&z_2\\0&1&z_3-\hat{k}\bra{i}\end{pmatrix}
\]
are singular.
Using $g(\hat{k}\bra{i})=0$ one checks that
\[
  v\bra{i}:=\big(\hat{k}\bra{i}(\hat{k}\bra{i}-z_3)-z_2,\,\hat{k}\bra{i}-z_3,\,1\big)\T
\]
is an eigenvector of~$C_g$ to the eigenvalue~$\hat{k}\bra{i}$.
By Proposition~\ref{P-Facts}(c), the points $v\bra{i}$ are on the curve~$\Gamma$ given by the equation $F(x_1,x_2,x_3)=0$ (see Notation~\ref{N-Notation}).
Furthermore, thanks to Proposition~\ref{P-ReducGamma} this polynomial factors over $\F_{q^3}$ into
\begin{equation}\label{e-F}
   F(x_1,x_2,x_3)=\prod_{i=0}^2(x_1+k\bra{i}x_2+\beta\bra{i}x_3).
\end{equation}
Without loss of generality let $v=v\bra{0}$ be on the line $x_1+kx_2+\beta x_3=0$.
The point~$v$ is also on the line
\begin{equation}\label{e-baseline}
    x_1+kx_2+\phi(k,\hat{k})x_3=0
\end{equation}
because $\phi(k,\hat{k})=(k+\hat{k})(\sigma_1(\hat{k})-\hat{k})-\sigma_2(\hat{k})=(k+\hat{k})(z_3-\hat{k})+z_2$ thanks to~\eqref{e-abczsigma}.
This implies $\beta=\phi(k,\hat{k})$.
Since $\phi(k\bra{i},\hat{k}\bra{i})=\phi(k,\hat{k})\bra{i}$ we thus have
\[
   F(x_1,x_2,x_3)=\prod_{i=0}^2\big(x_1+k\bra{i}x_2+\phi(k\bra{i},\hat{k}\bra{i})x_3\big).
\]
Thanks to Proposition~\ref{P-Facts}(a) the curve~$\Gamma$ has no points in $\P^2(q)$.
But this means that $1,k,\phi(k,\hat{k})$ are linearly independent over~$\F_q$.

Consider now
\[
  \det(xI-Z)=x^3-(z_2+z'_3)x^2-(z'_1+z_3z'_2-z_2z'_3)x-(z_1z'_2-z_2z'_1).
\]
Since~$Z$ is similar to $\text{diag}(\beta\bra{0},\beta\bra{1},\beta\bra{2})$, we obtain with the aid of~\eqref{e-sigmak} and~\eqref{e-abczsigma}
\begin{align*}
  \sigma_1(\beta)&=z_2+z'_3=z'_3-\sigma_2(\hat{k}),\\
-\sigma_2(\beta)&=z'_1+z_3z'_2-z_2z'_3=z'_1+\sigma_1(\hat{k})z'_2+\sigma_2(\hat{k})z'_3,\\
 \sigma_3(\beta)&=z_1z'_2-z_2z'_1=\sigma_3(\hat{k})z'_2+\sigma_2(\hat{k})z'_1.
\end{align*}
These are exactly the identities at the top of \cite[p.~292]{Men73}\footnote{In \cite{Men73} the quantities $c_{22}^h$ and $c_{21}^h$ play the role of our $z'_{h+1}$ and $z_{h+1}$,
respectively, for $h=0,1,2$, and $\alpha$ is our~$\beta$.
In the formula for $\sigma_3(\alpha)$ in the middle of p.~292 some terms are missing:  In the notation of~\cite{Men73} it has to read
$\sigma_3(\alpha)=\sigma_3(k+k')\sigma_3(c_{21}^2-k')-\sigma_2(\beta)\sigma_2(k')+\sigma_1(\beta)\sigma_2^2(k')-\sigma_2^3(k')$.},
and extended use of identities for the symmetric functions, derived in \cite[Lem.~10.3]{Men73}, leads to
\[
  z'_3=\sigma_1(k)\sigma_1(\hat{k})-\sigma_1(k\hat{k}),\quad
  z'_2=-\sigma_3(k+\hat{k}),\quad
  z'_1=\sigma_1(k)\sigma_3(\hat{k})+\sigma_1(\hat{k})\sigma_3(k)-\sigma_2(k\hat{k}),
\]
see \cite[pp.~292]{Men73}.
This means that $Z=\Sigma_2(k,\hat{k})$, as desired.
\\
``$\supseteq$''
We have to show that, whenever $1,k,\phi(k,\hat{k})$ are linearly independent, $(I,\,\Sigma_1(k,\hat{k}),\,\Sigma_2(k,\hat{k}))$ generates an MRD code.
This can be checked in two ways.
On the one hand, a very tedious computation (with the help of a CAS) shows that
\[
     \det\big(x_1I+x_2\Sigma_1(k,\hat{k})+x_3\Sigma_2(k,\hat{k})\big)=\prod_{i=0}^2\big(x_1+k\bra{i}x_2+\phi(k,\hat{k})\bra{i}x_3\big).
\]
Since $1,k,\phi(k,\hat{k})$ are by assumption linearly independent over~$\F_q$, the linear forms on the right hand side have no nonzero roots in $\F_q^3$ and thus the
triple $(I,\,\Sigma_1(k,\hat{k}),\,\Sigma_2(k,\hat{k}))$ generates an MRD code.
An alternative proof is given in~\cite{Men73} where a geometric argument is used in combination with more general versions of Propositions~\ref{P-Opp} --~\ref{P-ReducGamma}.
Those versions apply to more general triples $(I,C_f,Z)$, where the resulting $3$-dimensional algebras may have zero divisors, and lead to
a characterization of the void of zero divisors.
\hfill$\square$

We close this section with rewriting the characterizations of commutativity and associativity in Proposition~\ref{P-Fcommass} in terms of the parameters $k,\hat{k}$ from Theorem~\ref{T-MatSigma}.

\begin{prop}[\mbox{\cite[Prop.~11]{Men73}}]\label{P-Fcommass2}
Consider $(I,\Sigma_1(k,\hat{k}),\Sigma_2(k,\hat{k}))\in\cS$ as in Theorem~\ref{T-MatSigma}.
Thus $1,k,\phi(k,\hat{k})$ are linearly independent.
Let $a=\sigma_3(k),\,b=-\sigma_2(k),\,c=\sigma_1(k)$.
\begin{alphalist}
\item $\cF(I,\Sigma_1(k,\hat{k}),\Sigma_2(k,\hat{k}))$ is commutative iff $\hat{k}=k\bra{i}$ for some $i=0,1,2$.
        The case $i=0$ only occurs if ${\rm char}(\F_q)\neq2$.
        For $\hat{k}\in\{k\bra{1},\,k\bra{2}\}$ we have
        \[
            \Sigma_1(k,\hat{k})=\begin{pmatrix}0&0&a\\1&0&b\\0&1&c\end{pmatrix},\
             \Sigma_2(k,\hat{k})=\begin{pmatrix}0&a&ac\\0&b&a+bc\\1&c&b+c^2\end{pmatrix}=\Sigma_1(k,\hat{k})^2,
        \]
        whereas for $\hat{k}=k$ we have
                \[
            \Sigma_1(k,k)=\begin{pmatrix}0&0&a\\1&0&b\\0&1&c\end{pmatrix},\
             \Sigma_2(k,k)=\begin{pmatrix}0&a&4ac-b^2\\0&b&-8a\\1&c&-2b\end{pmatrix}.
        \]
\item $\cF(I,\Sigma_1(k,\hat{k}),\Sigma_2(k,\hat{k}))$ is associative (and thus commutative) iff $\hat{k}=k\bra{i}$ for some $i=1,2$.
\end{alphalist}
\end{prop}

The case $\hat{k}=k$, describing the unique commutative, non-associative semifield for a given~$k$ (in odd characteristic),
appears already in \cite[p.~374]{Dick06} by Dickson.
The case $\hat{k}\in\{k\bra{1},k\bra{2}\}$, given by the set~$\cS''$ in Corollary~\ref{C-SizeS}, parametrizes the semifields that are isomorphic to $\F_{q^3}$.
On the other hand, the set~$\cS'$ of the same corollary parameterizes the proper semifields (i.e., those that are not a field), including the commutative ones.
Proposition~\ref{P-CardSet} thus reflects the well-known fact~\cite[p.~208]{Knu65} that there are no proper semifields of order~$8$.

Finally note that Part~(a) above states that $\Sigma_j(k,k\bra{1})=\Sigma_j(k,k\bra{2})$ for $j=1,2$.
This agrees with the second option in Proposition~\ref{P-SigmaInj}(2).

\begin{proof}
(a) Using Proposition~\ref{P-Fcommass}(a) and~\eqref{e-sigmak} we obtain that $\cF(I,\Sigma_1(k,\hat{k}),\Sigma_2(k,\hat{k}))$ is commutative iff the minimal polynomials
of~$k$ and~$\hat{k}$ are identical.
This is the case iff $\hat{k}=k\bra{i}$ for some $i=0,1,2$.
It remains to see for which of these cases $1,k,\phi(k,\hat{k})$ are linearly independent.
We certainly need $k\in\F_{q^3}\setminus\F_q$.
From~\eqref{e-phikk} we know that  the elements are linearly independent if $\hat{k}\in\{k\bra{1},k\bra{2}\}$.
If $\hat{k}=k$, then $\phi(k,k)=2k(\sigma_1(k)-k)-\sigma_2(k)$.
Thus, if ${\rm char}(\F_q)\neq2$, then $1,k,\phi(k,k)$ are linearly independent because $\sigma_j(k)\in\F_q$ and $k\not\in\F_q$.
If, however, ${\rm char}(\F_q)=2$, then $\phi(k,k)\in\F_q$ and thus $1,k,\phi(k,\hat{k})$ are linearly dependent.
The rest of~(a) follows from the very definition of the matrices in Theorem~\ref{T-MatSigma} and some basic identities for
the symmetric functions, given in \cite[Lem.~10.3]{Men73}.
\\
(b) By Proposition~\ref{P-Fcommass}(b), $\cF(I,\Sigma_1(k,\hat{k}),\Sigma_2(k,\hat{k}))$ is associative iff
$\Sigma_1(k,\hat{k})^2=\Sigma_2(k,\hat{k})$.
Since commutativity follows from associativity for finite semifields, the statement is a consequence of~(a).
\end{proof}

\section*{Outlook}
It is natural to ask whether the methods can be extended to more general $[m\times n;\delta]$-MRD codes. 
To address this question let us recall where the particular case $(m,n,\delta)=(3,3,3)$ played a role.
Firstly, $\delta=n=m$ implies that (after some normalization) all nonzero matrices in the MRD code have no eigenvalues in~$\F_q$.
Secondly, since $n=3$,  the latter is equivalent to the irreducibility of their characteristic polynomials.
As a consequence, each such matrix is similar to a companion matrix.
Since $[3\times 3;3]$-MRD codes are $3$-dimensional, all of this has led to the set~$\cS$ consisting of matrix triples $(I,C_f,Z)$, and it remained to 
investigate which matrices~$Z$ lead to MRD codes.
This by itself was a major feat accomplished in Theorem~\ref{T-MatSigma} by Menichetti.
Let us briefly discuss the general square case of $[n\times n;n]$-MRD codes.
The first property given above (no eigenvalues in~$\F_q$) remains obviously true.
The second property (similarity to a companion matrix) needs to be replaced by general rational canonical forms (or other canonical forms).
While these two steps may just scale up in technicality, the main difficulty for this larger case, however, arises from the larger dimension of the MRD code.
A ``normalized'' basis may now be of the form $(I,C_{\rm rat.\,can.\,form},Z_1,\ldots,Z_{n-2})$, and one is left with the task of estimating the number of matrix tuples $(Z_1,\ldots,Z_{n-2})$ 
leading to an MRD code. 
It seems that entirely new methods are needed for this problem.
We thus close the paper with

\medskip
\noindent\textbf{Open Problem:}  For which $(m,n,\delta)$ are $[m\times n;\delta]$-MRD codes sparse?

\end{document}